\documentclass[a4paper,11pt]{article}
\usepackage{jheppub}
\usepackage{dsfont}
\usepackage{amsmath,amssymb,amscd,amsfonts,mathtools}
\usepackage{amsthm,thmtools}
\theoremstyle{plain}

\usepackage[toc,page]{appendix}
\usepackage{amsthm}
\usepackage{epsfig}
\usepackage{epstopdf}
\usepackage{latexsym}
\usepackage{graphicx}
\usepackage{placeins}
\usepackage{floatrow}
\usepackage{subfig}
\usepackage{caption}
\usepackage{booktabs}
\usepackage{bbm}
\usepackage{color}
\usepackage{physics}
\usepackage{stackrel}
\usepackage{tensor}
\usepackage{tikz}
\usetikzlibrary{matrix}
\usetikzlibrary{decorations.markings,calc,shapes,decorations.pathmorphing,patterns,decorations.pathreplacing}
\usetikzlibrary{positioning}
\usepackage{hyperref}
\newtheorem{theorem}{Theorem}

\pdfoutput=1
\makeatletter
\def\@fpheader{\relax}
\makeatother

\newtheorem*{theorem*}{Theorem}

\def\tilde#1{{\widetilde{#1}}}

\DeclareMathSizes{10}{10}{7}{7}

\subheader{\begin{flushright}
UT-WI-45-2023\\
\end{flushright}}

\title{Lorentzian threads and generalized complexity}

\author[a]{Elena C\'aceres,}
\author[b]{Rafael Carrasco,}
\author[a]{Vaishnavi Patil}
\affiliation[a]{Theory Group, Weinberg Institute, Department of Physics, University of Texas at Austin,\\
2515 Speedway, Austin, Texas 78712, USA.}
\affiliation[b]{Instituto de Física Teórica UAM/CSIC, Calle Nicolás Cabrera 13-15, Madrid 28049, Spain}
\emailAdd{elenac@utexas.edu}
\emailAdd{rafael.carrasco@ift.csic.es}
\emailAdd{vaishnavi.patil@utexas.edu}

\abstract{Recently, an infinite class of holographic generalized complexities was proposed. These gravitational observables display the behavior required to be  duals of complexity, in particular, linear growth at late times and switchback effect. In this work, we aim to understand generalized complexities in the framework of Lorentzian threads. We reformulate the problem in terms of thread distributions and measures and present a program to calculate the infinite family of codimension-one observables.   We also outline a
path to understand, using threads,  the more subtle case of codimension-zero observables.}
\begin{document}
\maketitle
\vfill
\pagebreak

\section{Introduction} \label{Introduction}
Undoubtedly, the Ryu-Takayanagi (RT)
proposal relating the entanglement entropy of a boundary region to the minimal area of a codimension-two bulk surface,  is a cornerstone in the connection between quantum information and quantum gravity. The RT proposal has become an essential tool in studying foundational aspects of holography  such as bulk reconstruction and the emergence and dynamics of spacetime. An alternative way to understand holographic entanglement entropy was proposed in \cite{Freedman:2016zud}. This formalism does not rely on bulk surfaces or minimal areas;  it uses convex optimization to identify the entanglement entropy as the maximization of a divergenceless norm-bounded vector flow. Specifically, at a constant time we have, 

 \begin{equation}\label{BitThreadReform}
S_A=\frac1{4G_N} \max_{v \in {\cal F}}\int_{A}\sqrt{h} \, n_\mu v^\mu\,, \qquad {\cal F}\equiv\{v \, \vert\, \nabla_\mu v^\mu=0,\, \abs{v}\leq 1\}.
\end{equation}
The integral lines of the flow are dubbed \emph{bit-threads} \cite{Freedman:2016zud}. This new language opened the possibility of understanding holographic constructs using convex optimization techniques. This approach has been a fruitful  avenue of research in past years. Bit threads have been used to study multipartite entanglement \cite{Freedman:2016zud, Harper:2019lff, Bao:2017nhh,Hubeny:2018bri,Agon:2018lwq}, holographic monogamy of mutual information \cite{Cui:2018dyq,Agon:2021tia}, the hypergraph entropy cone \cite{Bao:2020uku}, metric reconstruction \cite{Ag_n_2021}, etc.  Furthermore, recent developments extending the thread formulation to Lorentzian manifolds \cite{Pedraza:2021fgp} showed that quantities other than entanglement can be formulated in this language. We know now that \emph{entanglement is not enough} to describe the full spacetime, in particular the interior of black holes. Other constructs are necessary, for example, complexity. While in quantum information there are many notions of complexity,  holographic complexity is thought to be dual to circuit or gate complexity. This quantifies  the number of unitary operations that have to be applied to a reference state to get a target state within a tolerance.


The two main candidates are Complexity-Volume(CV) and Complexity-Action (CA). These two holographic constructs explore the interior of a black hole and have the properties expected from circuit complexity.
\begin{enumerate}
    \item Complexity-Volume: the volume of a maximal codimension-one Cauchy slice $\Sigma$ anchored at the boundary Cauchy slice $\sigma_A$ on which the CFT state is defined \cite{Susskind:2016,Susskind:2014jwa}.
    \begin{equation}
        C_V(\sigma_A) = \dfrac{1}{G_N l} \max_{\Sigma\sim A}\text{Vol} (\Sigma(A)).
    \end{equation}
    \item Complexity-Action: the bulk action evaluated on the Wheeler-De Witt patch associated to $\sigma_A$ \cite{Brown:2015bva,Brown:2015lvg}.
    \begin{equation}
        C_I (\sigma_A) = \dfrac{I_{WDW}}{\pi \hbar}.
    \end{equation}
\end{enumerate}
In  \cite{Headrick:2017ucz}, the authors derived a reformulation of the CV proposal in terms of flows by proving the max cut-min flow theorem.
Together with the bit threads construction in \cite{Freedman:2016zud}, this entailed developing the mathematical framework necessary to work with \emph{Lorentzian threads}, and allowed \cite{Pedraza:2021mkh} to formulate CV in terms of Lorentzian flows. In \cite{Pedraza:2021fgp}  the authors provided explicit constructions and an interlink between the discrete definition of complexity from the CFT perspective and the continuous definition of this quantity. More concretely, the integral lines of these flows correspond with ``gatelines", \textit{i.e.} timelike curves that represent unitary operations that transform a reference state to a target state.
Thus, there is a clear picture of how to understand CV in terms of threads. However, this understanding is missing for CA. Furthermore, recently, it was shown \cite{Belin:2021bga, Belin:2022xmt} that CV and CA are just two members of an infinite family of observables that exhibit the linear growth and switchback behavior of complexity. This freedom  in the definition of holographic complexity need not be perceived as a drawback; rather, it can be seen as a holographic feature that reflects the ambiguities existing in the definition of quantum complexity. However, much work needs to be done to understand this mapping. It is also possible that certain features will point to a particular type of holographic observable as the dual of complexity. Clearly, an understanding of the different facets and properties of this infinite class of gravitational observables is needed. In this work, we focus on one such facet. It is currently not known how to describe this infinite family of complexities, these \emph{generalized complexities}\footnote{We use the term "generalized complexities" to refer to the observables obtained from the ``complexity=anything" proposal of \cite{Belin:2021bga, Belin:2022xmt}. We adopt this name to emphasize that these observables are a generalization of  the CV, CA, and CV2.0 proposals.}, in terms of threads and optimization problems. In this paper we take a first step to remedy this situation. 

Let us briefly review the construction of generalized complexities. The generalized volume complexities \cite{Belin:2021bga} are  codimension-one, {\it i.e.} $d-$dimensional, observables \footnote{ Throughout this paper we  are working in a $d+1$ asymptotically $AdS$ space.} 
that involve two arbitrary functions, $F_1$ and $F_2$. The first step is to find the bulk surface $\Sigma$ anchored on a fixed  boundary time slice that maximizes the functional
\begin{equation}\label{eq:cg_functional}
   W_{F_2}= \int_\Sigma d^d \sigma \,\sqrt{h} F_2\,( g_{\mu \nu};X^\mu).
\end{equation}
Varying over the embeddings $X^\mu$ that define  $\Sigma$,  $\delta_X [W_{F_2}]=0$, we obtain the hypersurface $\tilde{\Sigma}$ that maximizes \eqref{eq:cg_functional}. Having determined $\tilde{\Sigma}$ we can evaluate the observable,
\begin{equation}\label{eq:gc_codim1}
\mathcal{O}_{F_1, \tilde{\Sigma}} =\frac{1}{G_N \ell} \int_{\tilde{\Sigma} }d^d \sigma \,\sqrt{h}\, F_1\,( g_{\mu \nu};X^\mu).
\end{equation}
If  $F_1=F_2=1$ we reproduce the standard maximal volume prescription, CV. While the family of observables obtained from \eqref{eq:gc_codim1} is large, it does not include CA or any observable\footnote{ Another common codimension-zero observable by the name of CV2.0 evaluates the spacetime volume of the WDW patch \cite{Couch_2017}} defined in a codimension-zero hypersurface.  The necessary extension of the construction was presented  in \cite{Belin:2022xmt} and involves six independent functions, $G_1, G_2, F_{1,+}, F_{1,-}, F_{2,+}$ and  $ F_{2,-}$. As before, we first determine the region that extremizes certain functional and then evaluate the observable in that region. More explicitly, consider the functional
\begin{align}\label{eq:gc_functional}
  W_{F_2\pm, G_2} = \int_{\Sigma_{+} }d^d \sigma \,\sqrt{h}\, F_{2,+}\,( g_{\mu \nu};X^\mu) &+\int_{\Sigma_{-} }d^d \sigma \,\sqrt{h}\, F_{2,-}\,( g_{\mu \nu};X^\mu) \nonumber\\
  & +\, \frac{1}{\ell} \int_{\mathcal{V} }d^{d +1}x \,\sqrt{-g}\, G_2\,( g_{\mu \nu};X^\mu),
\end{align}
where $\mathcal{V}$ is a region of spacetime bounded by $\Sigma_+$ and $\Sigma_-$.
Extremizing $ W_{F_2\pm, G_2}$ by varying the boundaries $\Sigma_{+}$ and $\Sigma_{-}$,  yields a  codimension-zero region of space, denoted $\tilde{\mathcal{V}}$ bounded by $\tilde \Sigma_{+}$ and $\tilde \Sigma_{-}$ where we evaluate the generalized observable,  

\begin{align}\label{eq: codim-0 observable}
\mathcal{O}_{F_1\pm, G_1, \Tilde{\Sigma}_{\pm}} =\frac{1}{G_N \ell} \int_{\tilde{\Sigma}_{+} }d^d \sigma \,\sqrt{h}\, F_{1,+}\,( g_{\mu \nu};X^\mu) +&\frac{1}{G_N \ell} \int_{\tilde{\Sigma}_{-} }d^d \sigma \,\sqrt{h}\, F_{1, -}\,( g_{\mu \nu};X^\mu)\nonumber \\+&\frac{1}{G_N \ell^2} \int_{\tilde{\mathcal{V}}}d^{d+1}x \,\sqrt{-g}\, G_1 \,( g_{\mu \nu};X^\mu) .
\end{align}
Note that \eqref{eq: codim-0 observable} defines an infinite class of gravitational observables that were shown \cite{Belin:2022xmt} to display the  universal
features required to be good candidates for holographic complexity. CA and CV2.0 are just two particular cases  that arise by choosing specific values of the functions $F_{2\pm},F_{1\pm}, G_1, G_2$ and  taking appropriate limits.  In this paper we aim to understand generalized complexities in the framework of threads. We find that it is useful to reformulate the problem in terms of thread distributions and measures. We present a program to calculate the infinite family of codimension-one observables   and outline a path to understand the more subtle case of codimension-zero observables. 

\section{Lorentzian threads and measures}\label{sec:Lorentzian_threads_and_measures}
In computational physics, circuit complexity is defined as the minimum number of unitary gates required to prepare a target state given an initial reference state within a defined tolerance. In a CFT, we expect that  the complexity of a state will grow over time due to Hamiltonian
 evolution. Holographically, this increase of complexity is conjectured to be  dual to the late-time growth of
the interior of a double-sided eternal AdS black hole. However, as reviewed in the Introduction, there is a large family of bulk constructions that are potentially dual to complexity. 
 The first and most studied proposals are Complexity-Volume (CV) and Complexity-Action (CA). 

In this section, we will review how Complexity-Volume (CV) can be understood in the language of flows and Lorentzian threads. This is conceptually similar to reformulating RT formula for entanglement entropy in the language of bit threads \cite{Freedman:2016zud}. After this, we will propose a new approach to understand Lorentzian threads in terms of measures. With this object, one would impose a constraint on the density of threads and propose a maximization program whose solution is in agreement with the CV proposal.

\subsection{Review of Lorentzian threads and CV}\label{sec: Review of Lorentzian threads and CV}


Before going into  the thread formulation for complexity, we will  establish some notation.  Let  $\mathcal{M}$  be a $(d+1)$-dimensional compact, oriented, Lorentzian manifold with boundary $\partial \mathcal{M}$. Consider a region $A$ on $\partial \mathcal{M}$ such that its causal future coincides with itself, that is, $J^+(A)=A$. This condition is imposed to guarantee the existence of a surface homologous to $A$. The boundary of $A$ will be denoted as $\sigma_A$,  $\partial A = \,\sigma_A$.  Thus, $\sigma_A$ is a codimension-two surface and as a consequence of its causal structure, $\sigma_A$ must be a Cauchy surface on the boundary manifold \cite{Headrick:2017ucz}. Let $\Sigma$ be a bulk codimension-one slice anchored on $\sigma_A$. We say that $\Sigma $  is homologous to $A$ ($\Sigma\sim A$) if there exists a bulk region $r(A)$ such that $\partial r \backslash \partial \mathcal{M} = -(\Sigma \backslash \partial \mathcal{M})$ (see figure \ref{general threads diagram} for an illustration).

 Given this setup, we can define a timelike flow $v$ 
and \textit{Lorentzian threads} as the integral lines of this vector field \cite{Headrick:2017ucz}. The timelike nature of this 
vector flow makes it necessary to reexamine some mathematical results previously established for this kind of spacelike flows. 
In \cite{Headrick:2017ucz}, the authors proved the Max Cut-Min Flow theorem for Lorentzian flows. Using this result, the volume of maximal slice $\tilde{\Sigma}$ anchored to $\sigma_A$ can be related to the minimal flux of Lorentzian flows through region $A$. Let us review this in more detail. 


A Lorentzian flow is defined as a divergenceless, future-directed vector field on $\mathcal{M}$ with a lower bound on its norm
\begin{equation}\label{eq: flow constraints}
    \nabla_{\mu} v^{\mu} = 0, \quad v^0 > 0, \quad |v| \geq 1.
\end{equation}
The flux through the boundary subregion $A$ is
\begin{equation}
    \int_A *v = \int_{\Sigma} *v = \int_{\Sigma} d^{d}\sigma\sqrt{h} n_{\mu} v^{\mu},
\end{equation}
where $*$ represents the Hodge dual, $n^{\mu}$ is the unit normal to $\Sigma$ and $\sqrt{h}$ is the induced volume element and $v$ is the one-form dual to the flow. The first equality follows from the divergenceless condition of the flow and second comes from the definition of flux. The boundedness of the vector field immediately gives a lower bound on the flux as
\begin{equation}
    \int_A *v \geq \text{Vol} (\tilde{\Sigma}).
\end{equation}

The max flow-min cut theorem states that, for the optimal flow, this inequality is saturated \cite{Headrick:2017ucz}, what implies
\begin{equation}
    \inf_v \int_A *v = \sup_{\Sigma \sim A} \text{Vol} (\Sigma).
\end{equation}


Thus, the maximal volume slice $\tilde{\Sigma}$ acts as an inverse bottleneck to minimum flux of the flow. The integral lines of a Lorentzian flow are referred to as \textit{Lorentzian threads} \cite{Pedraza:2021fgp}. Just like the Lorentzian vector flow, the threads are timelike and future-directed.  And the divergenceless condition imposed on the vector field guarantees that  the threads end on the boundary.  

Lorentzian threads provide a natural way of expressing  bulk geometric quantities related to CV. It is easy to prove properties like superadditivity of subregion complexity using flows and it provides new bounds on complexity as well \cite{Pedraza:2021fgp}. Heuristically,  we can think of these threads as ``gatelines" which are trajectories through spacetime that take a given reference state to the final target state. Complexity is then the minimum number of gatelines crossing the Cauchy slice on whose boundary the state is defined. This is conceptually similar to the circuit complexity idea of minimum number of unitary gates needed to prepare a final state.

\subsection{Complexity volume using  measures }\label{sec: CV -- using measures}

In the next section we will show that to understand generalized complexity in terms of Lorentzian threads it is convenient to reformulate the problem in terms of measures. As a warm up, let us explore the definition  of CV in this language. 

It might seem that, in order to solve the problem in terms of threads, one must determine the optimal flow and solve its integral lines. However, there is a better way to approach this problem. We can work directly with the set of all possible threads and propose a functional that has to be optimized. As a result of this program, the actual configuration of threads is found. With this object, it is first necessary to introduce the concept of a measure.

Given $\mathcal{F}$ a family of subsets of another set $\mathcal{P}$\footnote{For this family $\mathcal{F}$ to be well defined, it must include both the empty set $\emptyset$ and $\mathcal{P}$, and be closed under union and difference, \textit{i.e.}, if $A,B\in \mathcal{F}$ then $A\cup B\in \mathcal{F}$, $A\backslash B\in\mathcal{F}$ \cite{Bashkara1983}.},  a measure is a function $\mu:\mathcal{F}\rightarrow \mathbb{R}\cup\{-\infty,\infty\}$ that satisfies the two following conditions \cite{Bashkara1983}:

\begin{itemize}
    \item Acts trivially on the empty set, \textit{i.e.} $\mu(\emptyset)=0$.
    
    \item For a series of disjoints subset $F_n$ with $n\geq 1$ such that $\cup_{i\geq 1}F_i \subset \mathcal{F}$,
    \begin{equation}
        \displaystyle\nonumber \mu\left(\bigcup_{i\geq 1}F_i\right)= \sum_{i\geq 1} \mu(F_i).
    \end{equation}
\end{itemize}

We are interested in  $\mathcal{P}$ being  the set of all timelike and future directed curves going from $A^c=\partial \mathcal{M}\backslash A$ to $A$. These curves are just a representation of the Lorentzian threads introduced in the previous subsection. The measures we work with map a thread to the set $\{0,1\}$. In other words, a measure assigns a weight of $0$ or $1$ to each element in $\mathcal{P}$.

In order to propose a program that determines complexity-volume in terms of measures, we need to rewrite first the optimization program in terms of flows, given in \cite{Headrick:2017ucz}, using threads. In particular, both the objective and constraints must be translated into this new language. The divergenceless condition in \eqref{eq: flow constraints} is automatically implemented since, as was previously mentioned, threads cannot have one end in the bulk, but both ends have to lie in the boundary. The norm bound might be imposed by defining a delta function \cite{Headrick:2023}:

\begin{equation}
    \Delta(p,x)=\int_p ds \delta(x-y(s)),
\end{equation}

where $p\in \mathcal{P}$, $y(s)$ is the trajectory of the thread $p$ and $s$ is a proper parameter. This function counts the number of times a curve $p$ passes through the point $x$. Any element of $\mathcal{P}$ is by definition timelike and future directed and thus,  $\Delta(x,p)$ only takes the values $0$ or $1$. This function allows us to set a thread density $\rho(x)$ as

\begin{equation}
    \rho(x)=\int_\mathcal{P} d\mu \Delta(x,p),
\end{equation}

\noindent where the previous expression might be understood as a sum of the $\Delta$ function for each possible thread $p$ multiplied by the factor $\mu(p)$. This formula might be written as $\sum_{p\in \mathcal{P}}\mu(p)\Delta(x,p)$, but we do not employ the sum symbol since $\mathcal{P}$ is not a countable set. 

With the previous definition, the norm bound reads

\begin{equation}
    \rho(x)\geq 1,\ \forall x \in \mathcal{M}.
\end{equation}

Finally, we propose the objective to be given by the sum of weights of every single thread in $\mathcal{P}$. In this approach, the complexity of a CFT state on the slice $\sigma_A$ is then given by the following optimization program\footnote{We can set $G_N\ell=1$ and restore this factor when necessary.}

\begin{equation}
    \mathcal{C}=\min \frac{1}{G_N\ell}\int_{\mathcal{P}}d\mu,\ \ \textup{s.t. } \rho(x)\geq 1 \ \forall x\in \mathcal{M}.
    \label{eq: CV measure program}
\end{equation}

Solving \eqref{eq: CV measure program} is not an easy task. Nevertheless, provided that both the objective and constraint are linear functions of the measures (and hence, convex), the techniques of convex optimization will allow us to dualize it to a more tractable program. We refer the reader to \cite{Headrick:2017ucz,boyd_vandenberghe_2004} for a deeper insight into this topic. In particular, starting from \eqref{eq: CV measure program},  one can construct the following Lagrangian function (setting $G_N l = 1$)

\begin{equation}
\begin{split}
    L=&\int_{\mathcal{P}}d\mu+\int_\mathcal{M}d^{d+1}x \sqrt{-g}\lambda(x) \left(1-\rho(x)\right)\\
    =&\int_{\mathcal{P}}d\mu+\int_\mathcal{M}d^{d+1}x \sqrt{-g}\lambda(x) \left(1-\int_{\mathcal{P}}d\mu\Delta(x,p)\right),
\end{split}
\label{eq: primal lagrangian cv}
\end{equation}

\noindent where $\lambda(x)$ is a non-negative Lagrange multiplier\footnote{The non-negative character of this Lagrange multiplier comes from the fact that it is associated to an inequality constraint. If it were given for an equality constraint it may take any real value. More details about convex optimization can be found in \cite{Headrick:2017ucz,boyd_vandenberghe_2004}}. Note that at each $x \in \mathcal{M}$, there exists a constraint and hence a corresponding Lagrange multiplier. For this reason, an integral over all spacetime points of the function $\lambda(x)$ has to be performed. After rearranging the terms  in \eqref{eq: primal lagrangian cv}, we find

\begin{equation}\label{eq:dual_cv}
    L=\int_\mathcal{M}d^{d+1}x \sqrt{-g}\lambda(x) +\int_{\mathcal{P}}d\mu \left(1-\int_{p}ds\lambda(s)\right).
\end{equation}

\noindent For the Lagrangian to be lower bounded, we require the terms in the parenthesis to be non-negative. Written in this way, it is clear that \eqref{eq:dual_cv} be  interpreted as the Lagrangian of a \emph{dual} program, a program where the measure $\mu$ is  the Lagrange multiplier, 

\begin{equation}
    \max \int_\mathcal{M}d^{d+1}x \sqrt{-g}\lambda(x) ,\ \ \textup{s.t. } \int_{p}ds\lambda(s)\leq 1 \ \forall p\in \mathcal{P}.
    \label{eq: dual CV measure program}
\end{equation}

If the solution to \eqref{eq: CV measure program} is $\tilde{m}$ and that of \eqref{eq: dual CV measure program} $\tilde{p}$,  \textit{weak duality}  guarantees that $\tilde{p}\leq \tilde{m}$. But for the equality to hold, Slater's condition must be satisfied \cite{boyd_vandenberghe_2004}. 
Recall that Slater's condition is a necessary condition for strong duality. It states that, if the primal program \eqref{eq: CV measure program} admits a feasible solution $\tilde{\mu}$ (not necessarily optimal) which strictly satisfies the inequality constraints ($\tilde{\rho}(x)=\int_\mathcal{P} d\tilde{\mu} \Delta(x,p)>1$), then the optimal solutions of \eqref{eq: CV measure program} and \eqref{eq: dual CV measure program} coincide.

We will provide a heuristic argument showing that, indeed, Slater's condition is satisfied. One can cover the whole spacetime manifold with tubes of transverse area equal to one, even if they intersect. Inside each tube, one can insert finite (but larger than one) number of threads. As a consequence of this construction, the density of threads $\rho(x)$ is larger than $1$ for each point $x\in \mathcal{M}$ and, therefore, both programs share the same optimal solution. 

Our next goal is to show that the optimum of \eqref{eq: dual CV measure program} is the volume of the maximal slice homologous to $A$, which will be denoted $\tilde{\Sigma}$. Setting $\lambda(x)$ to be a delta function supported on $\tilde{\Sigma}$, one can realize that the constraints are satisfied. Thus

\begin{equation}
    \max \int_\mathcal{M}d^{d+1}x \sqrt{-g}\lambda(x)\geq \textup{Vol}(\tilde{\Sigma})
\end{equation}

In order to show that the equality in the previous expression holds, we will need the following theorem:\footnote{We provide the proof of this theorem in Appendix \ref{app:1}.}

\begin{restatable}{theorem}{thm}\label{th: theorem}
    Let $\mathcal{M}$ be a Lorentzian manifold, $A,\ B$ two complementary subsets  of the boundary such that $J^+(\partial A)|_{\partial \mathcal{M}}=B$ and $J^-(\partial B)|_{\partial \mathcal{M}}=A$, $\mathcal{P}$ the set of timelike, future directed curves going from $A$ to $B$ and $\lambda(x)$ a non-negative function on $\mathcal{M}$. Then, statements 1) and 2) are equivalent:
    \begin{align*}
       1)\,&\exists \psi:\mathcal{M}\rightarrow [-1/2,1/2]\ \ s.t.\ \psi|_A=-1/2, \psi|_B=1/2,\ \ |d\psi|\geq \lambda,\ d\psi\ timelike\ FD\\
      2)\, & \forall p\in \mathcal{P},\ \int_p ds\lambda\leq 1,\quad \text{with s  the proper distance along p.}
    \end{align*}
\end{restatable}


 As a consequence of Theorem \ref{th: theorem} we have,

\begin{equation}
    \max \int_\mathcal{M}d^{d+1}x \sqrt{-g}\lambda(x)\leq \int_\mathcal{M}d^{d+1}x \sqrt{-g}|d\psi|.
    \label{eq: ineq 1}
\end{equation}

 Following  \cite{Headrick:2017ucz}, we define the region

\begin{equation}
    r(c)=\{x\in \mathcal{M}|\psi(x)\geq c\},
\end{equation}
\noindent and its closure on the bulk $\Sigma(c)=\partial r(c)\backslash \partial M$. Due to the fact that $d\psi$ is future-directed and timelike, these $\Sigma(c)$ are slices\footnote{A slice is defined  \cite{Headrick:2017ucz} as a compact codimension-one hypersurface-with-boundary in $\mathcal{M}$ which is everywhere light or spacelike with future directed normal and whose interior is embedded in $\mathcal{M}$.}. Moreover, as $\psi|_A=-1/2$ and $\psi|_B=1/2$, they all must also be homologous to $A$, since all the level sets must have $\sigma_A=\partial A$ as boundary. Making use of the co-area formula \cite{Treude:2012np}

\begin{equation}
    \int_{\mathcal{M}} d^{d+1} x \sqrt{-g}|d\psi|=\int_{-1/2}^{1/2} dc \textup{Vol}(\Sigma(c))\leq \max_{\Sigma(c)\sim A} \textup{Vol}(\Sigma(c)).
\end{equation}

This last inequality, together with \eqref{eq: ineq 1}, finally shows that
\begin{equation}
   \mathcal{C}= \min \frac{1}{G_N\ell}\int_{\mathcal{P}}d\mu,\ \ \textup{s.t. } \rho(x)\geq 1 \ \forall x\in \mathcal{M}= \max_{\Sigma(c)\sim A} \textup{Vol}(\Sigma(c)).
\end{equation}

Thus, we have shown that using the language of measures, CV is obtained as a solution of the program in  \eqref{eq: CV measure program}.

\begin{figure}
    \centering
    \includegraphics[scale=0.3]{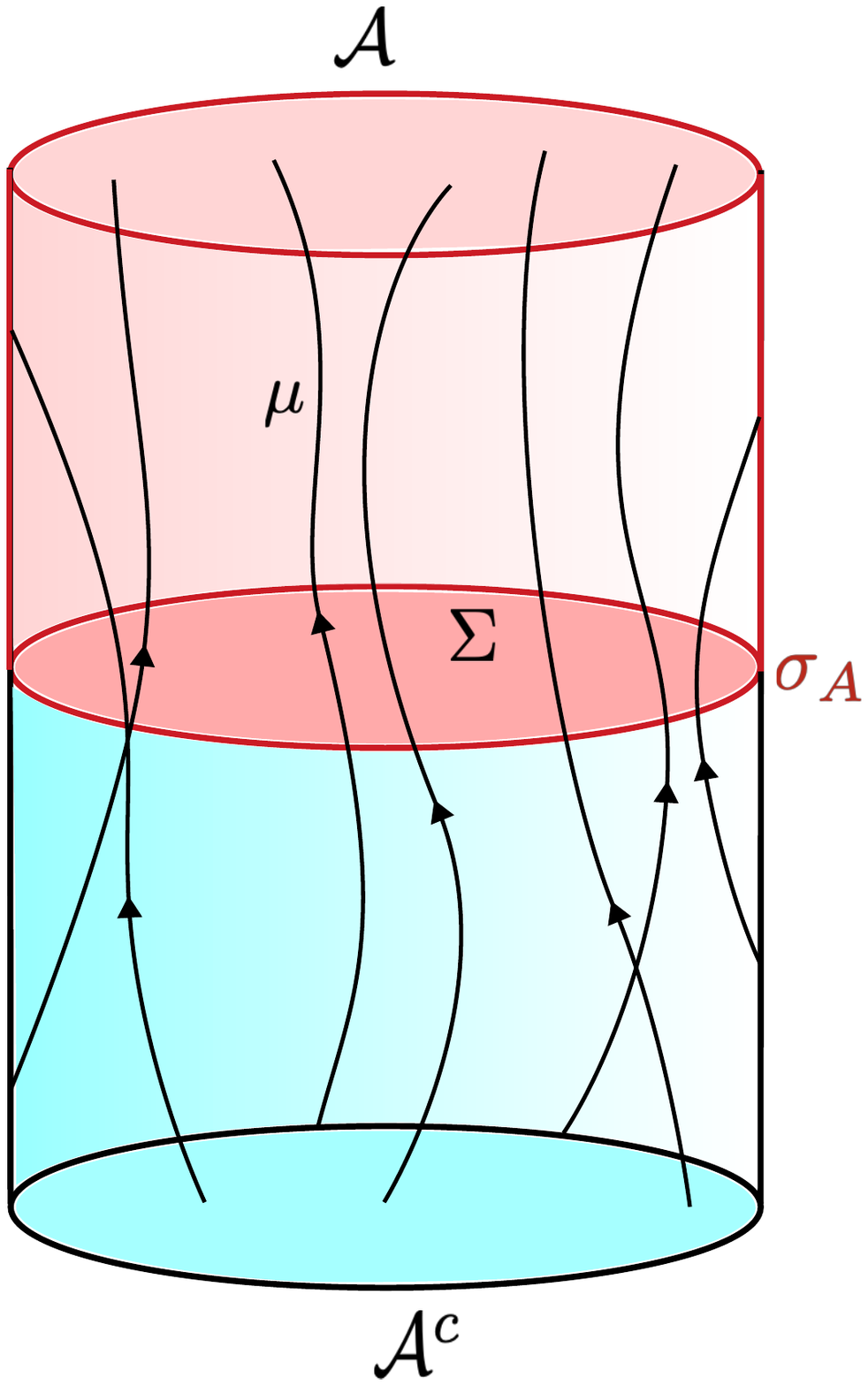}
    \caption{Threads with measure $\mu$ crossing the codimension-one surface $\Sigma$}
    \label{general threads diagram}
\end{figure}

\section{Threads and generalized complexity}\label{sec:threads_and_generalized_complexity}

In this section, we derive the main result of our work. We propose a program to find a Lorentzian thread configuration that calculates the generalized complexities of \cite{Belin:2021bga, Belin:2022xmt}. Specifically, we propose minimal measure threads with certain density constraints and apply the techniques of convex optimization to dualize the minimal program with constraints into a maximal program. 

\subsection{Codimension-one observables}\label{sec: codimension 1}

The observables described in \cite{Belin:2021bga} are calculated in two steps. The first step is to select a codimension-one surface $\tilde{\Sigma}$ that optimizes the integral of an arbitrary function $F_2(g_{\mu\nu},X^{\mu})$ on the surface. The next step is to evaluate the integral of a second arbitrary function $F_1(g_{\mu\nu},X^{\mu})$ on the same surface. 

To be a meaningful definition of circuit complexity, the final observable must be positive. However, the functions $F_1$ and $F_2$ may be positive or negative on $\tilde{\Sigma}$. We find that to find the optimal surface $\tilde{\Sigma}$ we only need positive measures but to evaluate the observable we need to introduce a rescaling that effectively ammounts to considering positive and negative measures.

Let us first focus on obtaining the optimal surface $\tilde{\Sigma}$ which extremizes the functional \eqref{eq:gc_codim1}. 
In section \ref{sec: CV -- using measures} we showed how to find the maximal volume hypersurface homologous to $\sigma_A$ from the thread formulation. 
Our goal here is to propose a program that will determine $\Sigma$ in the case of a generalized functional.
Unlike in CV, where we looked for a \emph{maximum} solution, in the present case we have an \emph{extremization} program. In principle, there may be more than one extremal surface. If that is the case, we should choose the one maximizing the complexity. 
However, we are still faced with the possibility of $F_2$ being negative.


If $F_2(x)$ is negative at certain codimension-zero subregions in spacetime, an extremal slice might have some patches lying inside these subregions. The functional will be maximal if, inside this patches, the maximal hypersurface is null, which is not a sensible result. So we use the absolute value of $F_2$ which turns an optimization program including both maxima and minima into a single maximisation program\footnote{One approach is to consider both positive and negative measured threads in the optimization program. However this does not work as we show in Appendix \ref{app B}}. In the appendix \ref{app:1}, we show the surfaces extremizing the functional $W(F_2)$ also extremize $W(|F_2|)$\footnote{In \cite{Chandra_2023} it was pointed out that the positivity of complexity required taking only positive functions $F_2$. This was implicitly done in \cite{Belin:2021bga,Belin:2022xmt}}. 
Nevertheless, positive measures cannot capture all the information of the sign of $F_2$ so, in order to do it, one solves the problem (just with positive measures) and then constructs the surface $\tilde{\Sigma}$. Then, those threads that intersect $\tilde{\Sigma}$ at a point x such that $F_2(x)$ is negative, must have their sign changed.

In contrast to the previous case, where the output of the optimization program was the volume of the maximal surface homologous to $\sigma_A$, we want a program whose solution is the maximum among the integral of $F_2(x)$ over all the surfaces homologous to this boundary region. To implement this, consider the following optimization program:
\begin{equation}
    \text{min} \quad \int_\mathcal{P}d\mu, \quad \textup{s.t.} \quad \rho(x)\geq |F_2(x)|,\quad \forall x \in \mathcal{M},
\label{C=any opt program}
\end{equation}
where $\mu(p)$ is non-negative and $F_2(x)$ is some function defined at every point in the manifold (which may take positive or negative values). The Lagrangian ($\lambda(x) \geq 0$) is:
\begin{equation}
    \begin{split}
        L[\mu, \lambda] & = \int_\mathcal{P}d\mu + \int_\mathcal{M} d^{d+1} x \sqrt{-g} \lambda(x) \left( |F_2(x)| - \int_{\mathcal{P}} d\mu\Delta (x,p) \right) \\
        & = \int_\mathcal{M} d^{d+1}x \sqrt{-g} \lambda(x) |F_2(x)| + \int_{\mathcal{P}} d \mu \left( 1 - \int_\mathcal{M} d^{d+1}x \sqrt{-g} \lambda(x) \Delta(x,p) \right) \\
        & = \int_\mathcal{M} d^{d+1}x \sqrt{-g} \lambda(x) |F_2(x)| + \int_{\mathcal{P}} d \mu \left( 1 - \int_p ds \lambda(y(s)) \right). \\
    \end{split}
\end{equation}

The dual program is then:
\begin{equation}
    \text{max} \int_\mathcal{M} d^{d+1}x \sqrt{-g} \lambda(x) |F_2(x)|,\quad \textup{s.t.} \quad \int_p ds\lambda\leq 1,\quad \forall p\in\mathcal{P}, 
\label{C=any dual program}
\end{equation}

\noindent where the inequality stems from requiring the minimum of the Lagrangian to be finite. 

The argument for Slater's condition to be satisfied is similar to the one given in the previous section for complexity-volume. Let us cover the whole spacetime with tubes of unit transverse area. We can find the point of the tube with the maximum value of $|F_2(x)|$, and fill the tube with a greater (but finite) number of threads. Hence the constraint at every point is strictly satisfied. This implies strong duality which means that the solution of \eqref{C=any opt program} and \eqref{C=any dual program} must coincide. 

In this case a feasible  $\lambda(x)$ is a delta function supported on a Cauchy slice $\Sigma$ with a fixed boundary. 
 This gives a lower bound on the solution of \eqref{C=any dual program}
\begin{equation}
    \max \int_\mathcal{M} d^{d+1}x \sqrt{-g} \lambda(x) |F_2(x)| \geq \max \int_{\Sigma}d^d \sigma \sqrt{h} |F_2(x)|.
\end{equation}
We can find the upper bound by using Theorem \ref{th: theorem} as follows
\begin{equation}
    \max \int_\mathcal{M} d^{d+1}x \sqrt{-g} \lambda(x) |F_2(x)| \leq \int_{\mathcal{M}}d^{d+1}x \sqrt{-g} |d\psi| |F_2(x)|,
\end{equation}
where $\psi : \mathcal{M} \rightarrow [-1/2,1/2] \  s.t. \ \psi|_{\mathcal{A}^c}=-1/2, \psi|_{\mathcal{A}}=1/2$ and $|d\psi| \geq \lambda(x)$. Using the coarea formula
\begin{equation}
    \int_{\mathcal{M}}d^{d+1}x \sqrt{-g} |d\psi| |F_2(x)| = \int_{-1/2}^{1/2} dc \int_{\Sigma(c)} d^d\sigma \sqrt{h} |F_2(x)| \leq \max \int_{\Sigma}d^d\sigma\sqrt{h} |F_2(x)|
\end{equation}
where $\Sigma(c)$ are level sets of $\psi$ homologous to $\mathcal{A}$.

Since the upper and lower limits of \eqref{C=any dual program} are same, it is the optimal solution. 
\begin{equation}\label{eq:intF2Sigma}
    \min \int_{\mathcal{P}}d\mu=\max \int_{\Sigma} d^d\sigma \sqrt{h}|F_2(x)|.
\end{equation}

Strong duality states that the maximum of the integral of $|F_2|$ on $\Sigma$ is equal to the minimum measure of the set of all threads with density constraints. 
Clearly, we can interpret this program as searching for the surface of maximum volume with respect to a re-scaled metric $\bar{g}_{\mu\nu}$ such that $\sqrt{-\bar{g}}=|F_2|\sqrt{-g}$ \footnote{The fact that the function $F_2(x)$ can be equal to zero at a certain region of spacetime might worry the reader. However, the nature of the maximization program will make the optimal surface not to intersect these regions. There are examples in which this sentence might not be true. For instance, suppose that inside the bulk domain of dependence of any Cauchy slice anchored to $\sigma_A$ there is a closed shell where $F_2=0$ while outside $F_2\neq 0$ but $F_2\ll 1$ and inside $F_2(x)\gg 1$. The optimal Cauchy surface will clearly cross the shell although $F_2(x)=0$ over it. }.
Notice that the program calculates the right-hand side of \eqref{eq:intF2Sigma} but does not provide the location of $\tilde{\Sigma}$. To obtain it, we need to resort to complementary slackness. It states that at every feasible point:
\begin{equation}
    \tilde{\lambda}(x)(F_2(x)-\tilde{\rho}(x))=0, \quad \forall x\in \mathcal{M}
\end{equation}
Therefore, for points at which $\lambda(x)\neq 0$ the inequality constraint gets saturated i.e. $\rho(x) = |F_2(x)|$. This allows us to reconstruct the surface $\tilde{\Sigma}$.
Having found the optimal surface $\tilde{\Sigma}$, the next step is evaluating the observable $O_{F_1,F_2}$. 

\subsubsection{Case I : $F_1(x) = F_2(x)$}
In the formulation of \cite{Belin:2021bga,Belin:2022xmt}, when $F_1(x) = F_2(x)$, finding $\tilde{\Sigma}$ and obtaining the value of the observable $\mathcal{O}_{F_1=F_2}$ is done in one step by maximizing \eqref{eq:gc_codim1}. In threads language, the situation is similar, except for a subtlety encountered if  $F_2$ is negative in some region. Recall that in section \ref{sec: codimension 1} we argued that to find $\tilde{\Sigma}$ we take $|F_2(x)|$ in equation \eqref{C=any opt program}. But to evaluate the observable we should take the sign of $F_2$ into account to reproduce the expected results \eqref{eq:gc_codim1}. We will see that this requires including negative measure threads when evaluating $\mathcal{O}_{F_1=F_2}$.

Consider the optimal configuration and the points $x$ where each thread crosses $\tilde{\Sigma},$ that is $x = p \cap \tilde{\Sigma}.$ Each of the threads that goes through a point where $F_2<0$ should contribute negatively to the observable. Thus, we must include positive and negative measure threads, 
\begin{equation}
\begin{split}
    & \forall x \in \Sigma \quad \textup{s.t.} \quad F_2(x) < 0, \quad \mu^{\prime}(p) = - \mu(p),\quad x = p \cap \tilde{\Sigma}, \\
    & \forall x \in \Sigma \quad \textup{s.t.} \quad F_2(x) \geq 0, \quad \mu^{\prime}(p) = \mu(p), \quad x = p \cap \tilde{\Sigma}, \\
\end{split}
\end{equation}
where $\mu(p)>0.$ The integral over  $\mu'(p),$ then gives the required generalized complexity,
\begin{equation}
   \mathcal{O}_{F_1=F_2}= \int_{\mathcal{P}} d\mu^{\prime} = \int_{\tilde{\Sigma}}d^d \sigma \sqrt{h} F_2(x).
\end{equation}

\subsubsection{Case II : $F_1(x) \neq F_2(x)$}\label{section: f_1 different f_2}
In section \ref{sec: codimension 1} we argued that  $\rho(x)=|F_2(x)|,\ \forall x\in \Sigma$. Thus, 
\begin{equation}
    \int_\tilde{\Sigma} d^d \sigma \sqrt{h}\rho(x)=\int_{\tilde{\Sigma}} d^d \sigma \sqrt{h}|F_2(x)|.
\end{equation}
It is convenient to discretize this integral in order to provide a cleaner geometrical interpretation. Suppose that threads intersect the surface $\tilde{\Sigma}$ at $m$ different points. Notice that $m\leq \#\textup{threads}$ since two or more threads can cross $\Sigma$ at the same point. We will partition the maximal surface into a set of smaller areas $A_i,\ i=1,...,m$ such that, 
$\int_{A_i}d^d\sigma\sqrt{h}|F_2(\sigma)|=n_i$
where $n_i$ is the number of threads going through $A_i$. 
Due to linearity of the integral,
\begin{equation}
 \int_{\tilde{\Sigma}} d^d \sigma \sqrt{h}|F_2(\sigma)|=\displaystyle \sum_{i} \int_{A_i} d^d \sigma \sqrt{h}|F_2(\sigma)|=\sum_{i} \mu(p_i),
\end{equation}
If every single $A_i$ has a small enough area\footnote{In a  holographic CFT the number of degrees of freedom, and therefore the number of threads, scales  like $N^2.$ Thus, each $A_i$ has a small area and the approximation \eqref{eq:approximation} holds .
}
\begin{equation}
    \int_{A_{i}} d^d\sigma \sqrt{h} |F_2(\sigma)|\approx \textup{ Area}(A_i) |F_2(x_i)|,
    \label{eq:approximation}
\end{equation}
where $x_i$ is the intersection of the thread $i$ with $\tilde{\Sigma}$. Thus, for threads appearing in the optimal configuration we have
\begin{equation}\label{eq: approx non intersecting}
    \mu(p_i)=\frac{1}{n_i}\textup{Area}(A_i) |F_2(x_i)|,
\end{equation}
where $A_i$ is the area enclosing $p_i$. In virtue of \eqref{eq:approximation} and \eqref{eq: approx non intersecting}, the observable  $\mathcal{O}_{F_1,F_2}$ can be obtained by rescaling the measure of each thread by $F_1(x_i)/|F_2(x)|$. Indeed,
\begin{equation}
   \begin{split}
   \mathcal{O}_{F_1,F_2}&= \displaystyle\sum_i \mu(p_i)\frac{F_1(x_i)}{|F_2(x_i)|}\approx \sum_{i=1}^{\# \textup{ threads}}  \frac{1}{n_i}\textup{Area}(A_i)F_1(x_i)\approx\sum_{i=1}^{\# \textup{ areas}}  \textup{Area}(A_i)F_1(x_i)\\
   &\approx \int_{\tilde{\Sigma}} d^d x \sqrt{h}F_1(x).
    \end{split}
\end{equation}
A couple of remarks are in order. First, note that the original measure  $\mu$ is strictly positive since it is the solution of maximizing $|F_2(x)|$. However, the multiplicative factor takes the sign of the function $F_1(x)$, and thus, the final observable also depends on the sign of the latter function. Second, $\mu$,  used to determine $\tilde{\Sigma}$,  has by construction a value equal to $0$ or $1$ on each thread. However, the rescaling makes the image of $\mu'$ to be the real numbers and not just the finite set $\{0,1\}.$ 


\subsection{Codimension-zero observables}
Once we have analysed in detail the thread formalism that allows the computation of generalized complexities, one can go one step further and not restrict just to programs that in the end yield a solution localized at a codimension-one surface, but it can be also extended immediately to codimension-zero surfaces, as shown in \cite{Belin:2022xmt}. The underlying idea followed in the latter paper is an extension of the functional \eqref{eq:gc_codim1} such that it is also evaluated on a codimension-zero region $\mathcal{V}$ in the bulk, as stated in \eqref{eq:gc_functional} that we recall here

\begin{equation}
    \begin{split}
        W_{G_2,F_{2,\pm}}(\mathcal{V})=&\int_{\Sigma_+} d^d\sigma \sqrt{h}F_{2,+}(g_{\mu\nu};X^\mu_+)+\int_{\Sigma_-} d^d\sigma \sqrt{h}F_{2,-}(g_{\mu\nu};X^\mu_-)\\
        &+\frac{1}{L}\int_{\mathcal{V}}d^{d+1}x\sqrt{-g}G_2(g_{\mu\nu}),
    \end{split}
\end{equation}

\noindent where $\Sigma_\pm$ are two codimension-one surfaces (with volume form $\sqrt{h}d^d\sigma$), $\mathcal{V}$ the region enclosed by these two surfaces and $X^\mu_\pm$ are the embeddings of $\Sigma_\pm$. Finding the previous surfaces reduces to an extremization of this functional under changes of $X^\mu_\pm$

\begin{equation}
    \delta_{X_{\pm}}\left[W_{G_2,F_{2,\pm}}(\tilde{\mathcal{V}})\right]=0.
\end{equation}

In general, the extremization of the functions is hindered by the appearance of the volume (or codimension-zero) term. One simplifying situation occurs when the function $G_2$ admits a primitive function $\tilde{G}_2$, in such a way that Stokes' theorem can be applied, the functional may be recast as

\begin{equation}
    \begin{split}
        W_{G_2,F_{2,\pm}}(\mathcal{V})=&\int_{\Sigma_+} d^d\sigma \sqrt{h}[F_{2,+}(g_{\mu\nu};X^\mu_+)+\tilde{G}_2(g_{\mu\nu};X^\mu_+)]\\
        &+\int_{\Sigma_-} d^d\sigma \sqrt{h}[F_{2,-}(g_{\mu\nu};X^\mu_-)-\tilde{G}_2(g_{\mu\nu};X^\mu_-)],
    \end{split}
\end{equation}  

and both $\Sigma_+$ and $\Sigma_-$ can be obtained from two different optimisation programs

\begin{equation}
    \begin{split}
        \delta_{X_+}\left(\int_{\tilde{\Sigma}_+} d^d\sigma \sqrt{h}[F_{2,+}(g_{\mu\nu};X^\mu_+)+\tilde{G}_2(g_{\mu\nu};X^\mu_+)]\right)=0,\\
        \delta_{X_-}\left(\int_{\tilde{\Sigma}_-} d^d\sigma \sqrt{h}[F_{2,-}(g_{\mu\nu};X^\mu_-)-\tilde{G}_2(g_{\mu\nu};X^\mu_-)]\right)=0.\\ 
    \end{split}
    \label{eq: codim-0 extremisation}
\end{equation}

Once $\Sigma_\pm$ are determined, the generalized complexity is given by

\begin{equation}
    \begin{split}
        &\mathcal{O}[G_1,F_{1,\pm},\mathcal{V}_{G_2,F_{2,\pm}}](\Sigma_{\textup{CFT}})=\frac{1}{G_N L}\int_{\tilde{\Sigma}_+[G_2,F_{2,+}]} d^d\sigma \sqrt{h}F_{1,+}(g_{\mu\nu};X_+^\mu)\\
        &+\frac{1}{G_N L}\int_{\tilde{\Sigma}_-[G_2,F_{2,-}]} d^d\sigma \sqrt{h}F_{1,-}(g_{\mu\nu};X_-^\mu)+\frac{1}{G_NL^2}\int_{\tilde{\mathcal{V}}[G_2,F_{2,\pm}]}d^{d+1}x \sqrt{-g}G_1(g_{\mu\nu}).
    \end{split}
\end{equation}

Notice that the extremizations in \eqref{eq: codim-0 extremisation} resemble much like that the one solved in Section \ref{sec: codimension 1}, fact that leads us to think that one can find these surfaces repeating exactly the same steps as before. First, we propose the following program to compute $\Sigma_+$

\begin{equation}\label{eq: min threads codim 0}
    \text{min} \quad \mu_+(P), \quad \textup{s.t.} \quad \forall x \in \mathcal{M} \quad \rho_+(x) \geq |F_{2,+}(x)+\tilde{G}_2(x)|.
\end{equation}

Its dual program is 
\begin{equation}
    \text{max} \int_{\mathcal{M}} d^dx \sqrt{-g} \lambda_+(x) |F_{2,+}(x)+\tilde{G}_2(x)|\quad \textup{s.t.} \quad \int_p ds\lambda_+(x)\leq 1\ \forall p\in\mathcal{P}_+. \label{max program}
\end{equation}

where $\mathcal{P}_+$ is the set of threads that go from $A^c$ to $A$. 
Although both $\mathcal{P}_+$ and $\mathcal{P}_-$ are exactly the same set, we put the labels $\pm$ to differentiate the families of threads that are employed to obtain $\tilde{\Sigma}_+$ from the homologous ones used to obtain $\tilde{\Sigma}_-$. Because of the same argument as before, the solution to this program is the integral of $|F_{2,+}(x)+\tilde{G}_2(x)|$ over $\tilde{\Sigma}_+$. 
 Similarly, one can obtain $\tilde{\Sigma}_-$ with an equivalent program using a measure $\mu_-(p)$. In this case, the primal program differs slightly from \eqref{eq: min threads codim 0} due to a change in the density constraint introduced to match with the second line in \eqref{eq: codim-0 extremisation}

\begin{equation}\label{eq: min threads codim 0 two}
    \text{min} \quad \mu_-(P),\quad \textup{s.t.} \quad \forall x \in \mathcal{M} \quad \rho_-(x) \geq |F_{2,-}(x)-\tilde{G}_2(x)|, 
\end{equation}

but the procedure is exactly the same which gives a dual as:
\begin{equation}
    \text{max} \int_{\mathcal{M}} d^dx \sqrt{-g} \lambda_-(x) |F_{2,-}(x)-\tilde{G}_2(x)|,\quad \textup{s.t.} \quad \int_p ds\lambda_-(x)\leq 1\ \forall p\in\mathcal{P}_-. \label{max program-2}
\end{equation}

It is a noteworthy fact that in this case we have two families of threads that do not interact among themselves, one associated to $\tilde{\Sigma}_+$ and another one to $\tilde{\Sigma}_-$ (see Fig \ref{two different measures}), each of them associated to the measures $\mu_+$ and $\mu_-$ respectively. Again, in order to reconstruct the surfaces from the optimal thread configuration, it is only necessary to search those points at which $\rho_+(x)$ or $\rho_-(x)$ saturate the constraints, since this set of points constitute the support of the Lagrange multipliers $\lambda_+(x)$ and $\lambda_-(x)$. 

\begin{figure}
    \centering
    \includegraphics[scale=0.3]{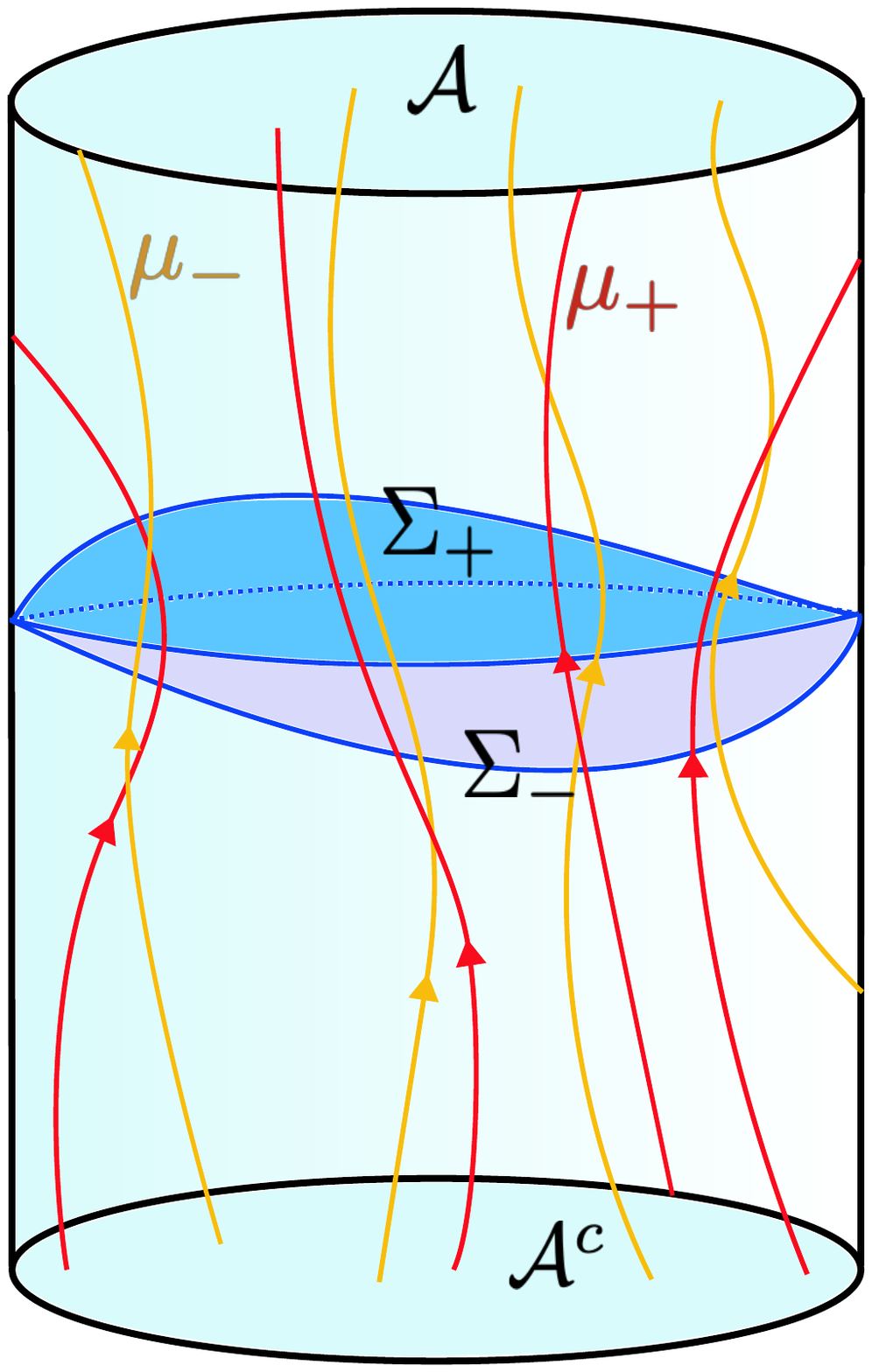}
    \caption{Two different measures $\mu_+$ and $\mu_-$ at each $x \in \mathcal{M}$ }
    \label{two different measures}
\end{figure}

As done in the previous sections, the optimization program would give the observable in the case of $F_{1,\pm}(x)=F_{2,\pm}(x), G_1(x)=G_2(x)$ (non-negative). In contrast, when $F_{1,\pm}(x)\neq F_{2,\pm}(x), G_1(x) \neq G_2(x)$, we have to perform a rescaling of the measures analogous to that shown in Section \ref{section: f_1 different f_2}. More concretely, complexity would be given by

\begin{equation}\label{eq: observable codim-0}
\begin{split}
    \mathcal{O}_{F_{1,\pm},F_{2,\pm},G_1,G_2} &= \displaystyle\sum_{i=1}^{\# \mathcal{P}_+} \mu_+(p_i)\frac{F_{1,+}(x_i)+ \Tilde{G}_1(x_i)}{|F_{2,+}(x_i) + \Tilde{G}_2(x_i)|}+\sum_{i=1}^{\# \mathcal{P}_-} \mu_-(p_i)\frac{F_{1,-}(x_i) - \Tilde{G}_1(x_i)}{|F_{2,-}(x_i)-\Tilde{G}_2(x_i)|} \\
    & =\sum_{i=1}^{\# \mathcal{P}_+} \mu'_+(p_i)+\sum_{i=1}^{\# \mathcal{P}_-} \mu'_-(p_i),
\end{split}    
\end{equation}

where in the last equality, we have made the redefinition $\mu'_\pm=\mu_\pm\frac{F_{1,\pm}(x_i) \pm \Tilde{G}_1(x_i)}{|F_{2,\pm}(x_i) \pm \Tilde{G}_2(x_i)|}$. We have again made the assumption here that the function $G_1$ has a primitive $\Tilde{G}_2$ so that the integral over a codimension-zero region can be reduced to an integral over the codimenion-1 surfaces. In the next section we discuss the cases where this is not possible.

\subsubsection{CV 2.0 and CA}

In general, the computation of observables evaluated on codimension-zero surfaces is hard to execute. The main reason is found in the fact that our formalism is highly based upon the assumption $G_2(x)\sqrt{-g}\dd^{d+1} x=\dd(\tilde{G}_2(x)\sqrt{h}\dd^d\sigma)$, where $\sqrt{-g}\dd^{d+1}x$ is the volume form of the spacetime while $\sqrt{h}\dd^{d}\sigma_\pm$ is corresponding to the surface $\tilde{\Sigma}_\pm$. Even in the simplest cases, as for instance, when $G_2(x)$ is constant, finding $\tilde{G}_2(x)$ is a highly non-trivial exercise whose solution is not formally guaranteed.

As in \cite{Belin:2022xmt} we will assume the existence of a $\tilde{G}_2(x)$ such that the associated $G_2(x)$ is constant and take $F_{1\pm}=F_{2\pm}=\alpha_{\pm}\rightarrow 0$. In \cite{Belin:2022xmt} it is shown the solutions to this problem are two surfaces $\tilde{\Sigma}_\pm$ such that $\tilde{\Sigma}_+\cup\tilde{\Sigma}_-$ is the boundary of the Wheeler-de Witt patch. 

Once the surfaces are known, the observable $\mathcal{O}_{F_{1\pm},F_{2\pm},G_1,G_2}$ would be determined after integrating the function $G_1$ in the region delimited by these two surfaces and $F_{1+}$ and $F_{1-}$ over $\tilde{\Sigma}_\pm$ respectively. The contribution of the latter surfaces is determined following the same approach as in \ref{section: f_1 different f_2}. Nonetheless, the piece evaluated on the codimension zero surface cannot be extracted from the threads already presented and finding a new way of perform this calculation becomes a central objective. Note that this is conceptually different from what has been previously done in the literature since  until now the structure of the threads away from the optimizing surface played no role in the observables. However, in the case of codimension-zero observables, the threads themselves are contained in the region. Here we will take a pragmatic approach; We  motivate how one can perform the calculation of the observable and leave questions about the interpretation and the exploration of other possible methods for future work.   

First note that the trace of the extrinsic curvature tensor can be employed to parametrize the set of hypersurfaces of constant $K$ homologous to the boundary region $\sigma_A$ \cite{MARSDEN1980109}. Let us select a finite but large set of $N+1$ of them. This set will be denoted with $D={\Sigma_0,\Sigma_1,...,\Sigma_N}$ where $\Sigma_0=\tilde{\Sigma}_-$, $\Sigma_N=\tilde{\Sigma}_+$. In particular, we want two contiguous surfaces to be very close one to the other. We look at the intersection points between threads and $\Sigma_0$ and construct a line perpendicular to this surface with constant tangent vector. We illustrate this construction in Fig \ref{fig:sewing}. In case several threads intersect at the same point in $\tilde{\Sigma}_+$, we just build one line. These lines intersect $\Sigma_1$ at another set of points. We can repeat the same procedure until the space between $\tilde{\Sigma}_+$ and $\tilde{\Sigma}_-$ is sewed with a set of threads. Then, we find a tessellation of each hypersurface in $D$. 
\begin{figure}
    \centering
    \includegraphics[scale=0.3]{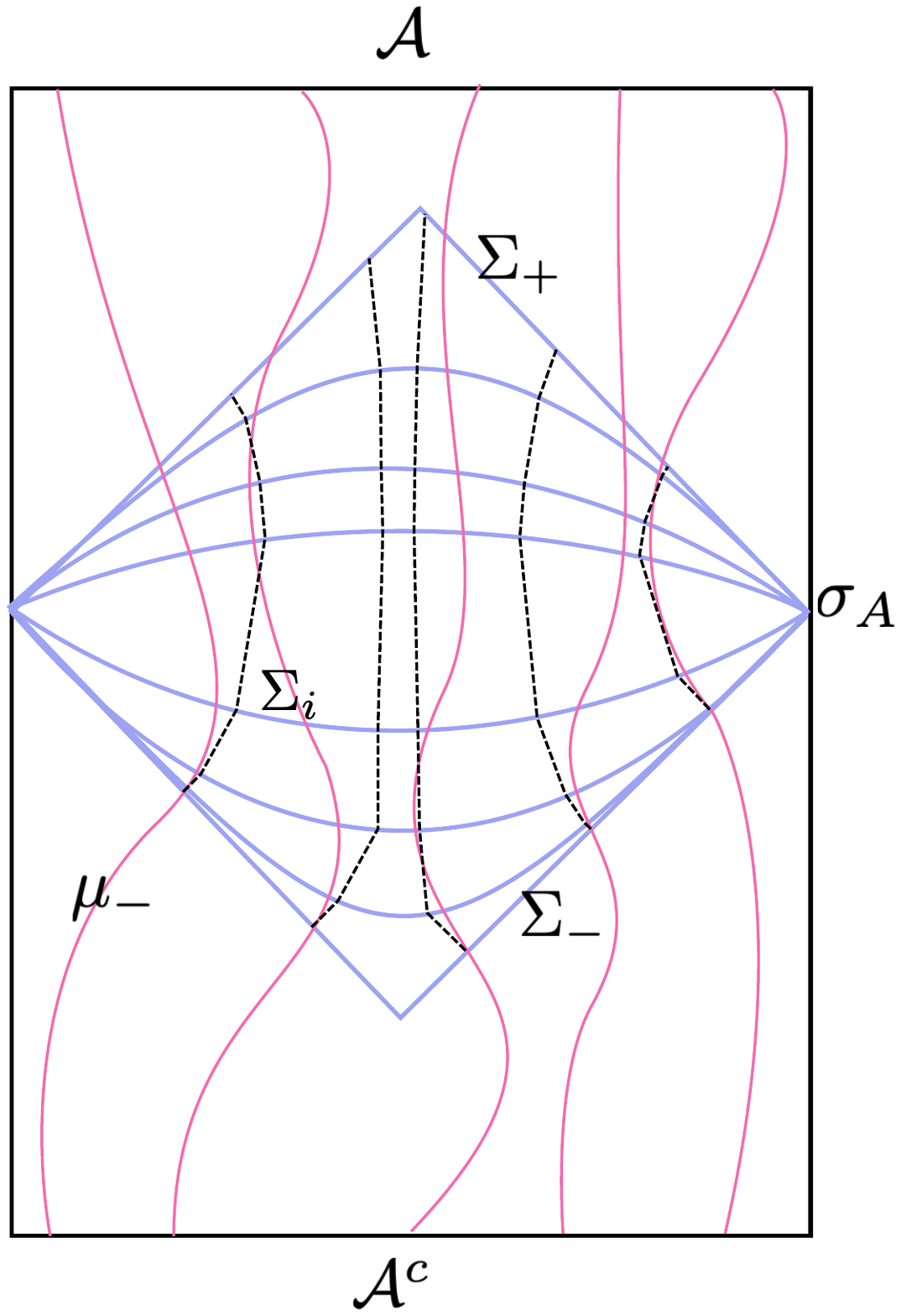}
    \caption{Foliating spacetime with constant curvature slices and sewing them with threads}
    \label{fig:sewing}
\end{figure}

We will label each area $A_{ij}$ with two indices, where the first index identifies the hypersurface and the second one labels the infinitesimal area element in it. If there are $n$ threads, the observable would approximately be given by

\begin{equation}
    \mathcal{O}_{F_{1\pm},F_{2\pm},G_1,G_2}=\int_\mathcal{V}d^{d+1}xG_1(x)\approx \displaystyle \sum_{i=0}^{N-1}\sum_{j=1}^{n} A_{ij}\int_{s_{i}}^{s_{i+1}}ds_j G_1(x)
\end{equation}

where $s_j$ is employed to parametrize the thread $j$ and $s_i$ is the value of this parameter when the thread intersects the hypersurface $\Sigma_i$.


To see the specific cases of interest we set $G_1(x) = 1, \; \forall x \in \mathcal{M}$ and see that the observable becomes the spacetime volume of the Wheeler-deWitt patch, thus giving us the CV 2.0 result. The less trivial case is found by setting $G_2(x)$ to be the Lagrangian density at each point $\mathcal{L}(x)$. Then the observable becomes the action of the Wheeler-deWitt patch as required by Complexity Action.

Note that our construction picks the points from one set of threads which saturate the surface $\tilde{\Sigma}_-$ and builds up infinitesimal area and volume elements from that. 
Alternatively, we could have started from the set of threads that saturates $\tilde{\Sigma}_+$ and built this structure in the past direction. For discrete set of threads, these two methods will give different solutions. However as we go to the continuum limit, the answers will converge to the same value.

\section{Conclusions}

In \cite{Belin:2021bga,Belin:2022xmt} the authors proposed an infinite family of generalized holographic complexities. The existence of  a large number of  gravitational observables that reproduce  the behavior necessary to be  dual of complexity, underscores the fact that our holographic understanding of quantum complexity still needs to be developed much further.

In the past, bit threads have  provided an intuitive understanding of holographic entanglement. Therefore, it is natural to ask if Lorentzian threads can provide a formulation of generalized complexities that would shed some light on their interpretation. In that spirit, in this paper we initiate the study of  generalized complexities in the language of Lorentzian threads.  For the standard maximal volume complexity, CV, the work of  \cite{Pedraza:2021fgp} proposes the interpretation of Lorentzian threads as ``gatelines" where each thread going from the complementary past of a boundary subregion to the subregion itself represents one gate of unitary evolution.  We find that to understand generalized complexitites it useful to reformulate the problem using the language of distributions and measures \cite{Headrick:2023} instead of flows as done in \cite{Pedraza:2021fgp}. We encountered subtleties associated with the fact that the generalized observables are defined by functions that can be positive or negative in different regions. We find  it  necessary to introduce threads with negative contributions or negative measures. With all these ingredients, we find an optimization program to calculate the infinite family of codimension-one, or generalized volume, observables.  We also propose a way of tackling the codimension-zero observables, although much work needs to be done in this case. Our works opens several directions of research: 

\begin{itemize}
    \item \emph{ Thread classes.} Can we identify subclasses of threads distributions or equivalently, complexity observables  by demanding they satisfy other properties?  For example, we can ask if there are particular thread configurations that encode Einstein's equations \cite{Pedraza:2022dqi,Carrasco:2023fcj, EJARV}. Or investigate issues of unitarity in time asymmetric situations as was done in \cite{Aguilar-Gutierrez:2023zqm}.
    
    \item \emph{ Explicit flow constructions.}  The  methods developed in  \cite{Pedraza:2021fgp} could be useful to find explicit flow constructions for different observables. Furthermore, our approach is agnostic regarding the existence of a horizon. Thus,  these constructions would allow us to ask questions regarding the interior of the black holes in terms of threads in the spirit of \cite{J_rstad_2023}.
    
    \item \emph{Interpretation of threads beyond CV.} It has been proposed \cite{Pedraza:2021fgp} that for the maximal volume prescription, CV,  the threads represent gatelines. In the generalized complexity framework, CV corresponds to $|F_1|=|F_2|=1$. The interpretation of threads in a more general case is an open question. A  concrete and necessary step could be understanding the  interpretation of negative measures needed to  formulate generalized complexities. 
   
\end{itemize}

\acknowledgments 

We are indebted to Jos\'e Barb\'on, Jes\'us Gamb\'in, Juan Pedraza and Andy Svesko for enlightening and fun  discussions. Special thanks to Juan Pedraza for comments on the manuscript.  EC thanks the Instituto de F\'isica Te\'orica UAM-CSIC (IFT UAM-CSIC) 
, for hospitality during the last stages of this project. EC and RC thank the organizers and participants of the "Complexity in Field Theory and Gravity" workshop at IFT UAM-CSIC for a stimulating conference. 
The work of RC is supported through the grants CEX2020-001007-S and PID2021-123017NB-I00, funded by MCIN/AEI/10.13039/501100011033 and by ERDF A way of making Europe. EC and VP are supported by the National
Science Foundation under grants No. PHY-2112725 and No. PHY–2210562.
\begin{appendices}

\section{ Proofs of Theorem 1 and Theorem 2}\label{app:1} 
In this Appendix we will prove the theorems used in section \ref{sec: CV -- using measures} and \ref{sec: codimension 1}.

\thm*

\begin{proof}
We will follow a similar derivation as in \cite{Headrick:2023}. 
Let us assume that 1) holds. Therefore one finds:

\begin{equation}
    1=\psi|_A^B=\int_p d\psi =\int_p ds \frac{d\psi}{ds}=\int_p ds \frac{d y^\mu}{ds}\partial_\mu \psi=\int_p ds \frac{d y^\mu}{ds}\partial^\nu \psi g_{\mu\nu}
\end{equation}

where $y^\mu$ is the trajectory of the thread. Before going one step further, let us take into account that $d\psi$ is timelike and future directed. The vector field dual to this one ($\partial^\mu \psi\partial_\mu$) will be also timelike but past directed. As the thread is directed to the future, the vector field $\frac{dy^\mu}{ds}\partial_\mu$ is future directed. It is easy to show that, under these circumstances, $g_{\mu\nu}u^\mu v^\nu\geq |u||v|$ for $u$ timelike, future directed and $v$ timelike, past directed. Hence 

\begin{equation}
    1\geq \int_p ds\left|\frac{d y^\mu}{ds}\right||d\psi| 
\end{equation}

Provided that, by definition of proper distance, $\left|\frac{dy^\mu}{ds}\right|=1$, one finds

\begin{equation}
    1\geq \int_p ds|d\psi| \geq \int_p ds\lambda
\end{equation}

In other words, $\int_p ds \lambda \leq 1$, reaching in this way 2).

Now, we have to show the other implication. Let us assume 2) and rewrite the integral over the thread as $\int_p ds \lambda=\int_p dt |-\dot{x}|\lambda$. In this expression $-\dot{x}$ is the covector associated the tangent to the curve $p$. The reason why we add a minus sign is that the tangent vector is future directed. We define

\begin{equation}
    \psi_-(y):=\sup_{\substack{p\ \textup{timelike}\\\textup{from } A\ \textup{to}\ y}}\int_p dt |-\dot{x}|\lambda,\quad  \psi_+(y):=\sup_{\substack{p\ \textup{timelike}\\\textup{from } y\ \textup{to}\ B}}\int_p dt |-\dot{x}|\lambda. 
\end{equation}

where the supremum is over all the curves going from $A$ to $y$ and from $y$ to $B$ respectively. By assumption

\begin{equation}
    \psi_-(y)+\psi_+\leq 1
\end{equation}

and also:

\begin{equation}
    \lim_{y\rightarrow A} \psi_-(y)=0,\ \ \  \lim_{y\rightarrow B} \psi_+(y)=0
\end{equation}

Let us now calculate the gradient of $\phi_-$. Notice that the integrand is a differentiable function of $\dot{x}$ since it corresponds to timelike future directed curves. To be able to apply the Hamilton-Jacobi formula\footnote{Hamilton-Jacobi formula assures that the variation of the on-shell action after a modification in the final position is given by the canonical momentum in its final value.} \cite{Headrick:2023}, we have to remark that the optimal solution may have a lightlike tangent vector. In order to overcome this problem, we will consider the domain to be the whole tangent space and implement the constraint by defining the integrand to equal $-\infty$ whenever the velocity is outside the future light cone.

\begin{displaymath}
\psi_-(y)=\sup_{\substack{q\ \textup{timelike}\\\textup{from } A\ \textup{to}\ y}}\int_q dt\left\{ \begin{array}{ll}
|-\dot{x}|\lambda & -\dot{x}\in \mathfrak{i}^+\\
-\infty &  \textup{otherwise} \\
\end{array}
\right.,
\end{displaymath}

\begin{displaymath}
\psi_+(y)=\sup_{\substack{q\ \textup{timelike}\\\textup{from } y\ \textup{to}\ B}}\int_q dt\left\{ \begin{array}{ll}
|-\dot{x}|\lambda & -\dot{x}\in \mathfrak{i}^+\\
-\infty &  \textup{otherwise} \\
\end{array}
\right. .
\end{displaymath}

Here, $\mathfrak{i}^+$ is the set of non-spacelike, future directed one-forms. If $-\dot{x}$ is timelike (which is the case of interest for us), we find that $\pi_{\pm\mu}=\partial_{\dot{x}^\mu}(|-\dot{x}|\lambda)=-\lambda \dot{x}_{\mu}/|-\dot{x}|$. Consequently

\begin{equation}
    |d\psi_\pm|^2=\frac{\lambda^2 \dot{x}_\mu\dot{x}^\mu}{|-\dot{x}|^2}\geq \lambda^2
    \label{eq: dpsi geq lambda}
\end{equation}

Therefore $|d\psi_\pm|\geq \lambda$. Moreover, setting 

\begin{equation}
    \psi(y)=\frac{\psi_-(y)-\psi_+(y)}{2(\psi_-(y)+\psi_+(y))},
\end{equation}

one may immediately realize that $\psi|_A=-1/2$ and $\psi_B=1/2$ and, taking into account \eqref{eq: dpsi geq lambda}

\begin{equation}
    |d\psi|\geq\frac{\psi_+(y)}{(\psi_-(y)+\psi_+(y))^2}|d\psi_-|+\frac{\psi_-(y)}{(\psi_-(y)+\psi_+(y))^2}|d\psi_+|\geq \frac{1}{\psi_+(y)+\psi_-(y)}\lambda\geq \lambda.
\end{equation}

Proving in this way that 2) implies 1).
\end{proof}

We now proceed with Theorem \ref{th: theorem2} and that we recall here 

\begin{theorem}\label{th: theorem2}
    Let $\mathcal{M}$ be a $d+1$-dimensional Lorentzian manifold, $\sigma$ a Cauchy slice living on the boundary of $\mathcal{M}$ and $\Sigma$ a codimension-one Cauchy surface whose boundary coincides with $\sigma$. The extremization of the functional

    \begin{equation}\label{eq: extremisation}
        W=\int_{\Sigma} d^{d}x\sqrt{h} F_2,
    \end{equation}
    under changes of the integration manifold $\Sigma$ provides the same extrema as the extremization of the functional

    \begin{equation}
       W'= \int_{\Sigma} d^{d}x\sqrt{h} |F_2|.
    \end{equation}
\end{theorem}

\begin{proof}
We first remark that the variation of the position of the surface $\Sigma$ in our case cannot be interpreted as a modification of the background metric $g_{\mu\nu}$ but is given by a modification of the embedding coordinates $X^\mu$ (do not confuse with $x^\mu$, which represent the spacetime coordinates). Furthermore, apart from the surface, the function $F_2$ will also have an explicit dependence in these embedding coordinates. The procedure we are going to follow is analogous to that shown in Appendix B of \cite{Feng:2017ygy}. 

For simplicity, we will write the volume form of $\Sigma$ as $d\Omega=d^d x\sqrt{h}$ and taking into account that the variation of this form generated by a vector field $\delta X$ (corresponding to the variation of the position of the surface)

\begin{equation}
    \delta d\Omega=\pounds_{\delta X} d\Omega.
\end{equation}

The displacement vector $\delta X$ might be decomposed as the sum of two terms: one normal to $\Sigma$ and another tangent to the latter surface. That is \cite{Feng:2017ygy}

\begin{equation}
\begin{split}
    \delta X^\mu&=\delta a n^\mu +\delta b^\mu,\\
    \delta a&= - n_\mu \delta X^\mu,\\
    \delta b^\mu&= h_\nu^\mu \delta X^\nu,
\end{split}
\end{equation}

where $n^\alpha$ is the vector normal to $\Sigma$. In can be shown that, $\delta d\Omega$ be written as \cite{Feng:2017ygy}

\begin{equation}
    \delta d\Omega=(D_i\delta b^i+\delta a K)d\Omega,
\end{equation}
where $K$ is the extrinsic curvature of $\Sigma$, $D_i$ is the covariant derivative induced on it and $i$ runs from $1$ to $d$. Hence, the variation of $\delta W$ reads

\begin{equation}
    \delta W=\int_\Sigma \left(\delta d\Omega F_2+d\Omega \frac{\partial F_2}{\partial X^\mu} \delta X^\mu\right).
\end{equation}




Applying integration by parts and Stokes theorem

\begin{equation}
    \delta W = \int_\Sigma \left[\delta a K F_2 +\frac{\partial F_2}{\partial X^\mu}(\delta b^\mu+n^\mu \delta a)-D_\mu F_2\delta b^\mu \right]d\Omega+\int_{\partial \Sigma} \sqrt{|\gamma|} F_2\delta b^i d\sigma_i
\end{equation}

where $\gamma$ is the deteminant of the induced metric on $\partial \Sigma$. Keeping in mind that we do not vary of the boundary's surface, the last term just vanish. On the other hand, the extremality condition imposes that $\delta W=0$ for arbitrary variations $\delta a$ and $\delta b^i$ so

\begin{equation}
\left\{ \begin{array}{l}
\left[KF_2+\frac{\partial F_2}{\partial X^\mu}n^\mu\right]_{X^\mu=\tilde{X}^{\mu}}=0\\
\left[\frac{\partial F_2}{\partial X^\mu}-D_\mu F_2\right]_{X^\mu=\tilde{X}^{\mu}}=0\\
\end{array}\right. \quad \forall \tilde{X}\in \Sigma.
\end{equation}

In particular, if there is a point $\tilde{X}\in \Sigma$ such that $F_2(\tilde{X})\geq 0$ it is immediate that 

\begin{equation}
\left\{ \begin{array}{l}
\left[K|F_2|+K\frac{\partial |F_2|}{\partial X^\mu}n^\mu\right]_{X^\mu=\tilde{X}^{\mu}}=0\\
\left[\frac{\partial |F_2|}{\partial X^\mu}-D_\mu |F_2|\right]_{X^\mu=\tilde{X}^{\mu}}=0\\
\end{array}\right. 
\end{equation}

or if $F_2(\tilde{X})\leq 0$ 

\begin{equation}
\left\{ \begin{array}{l}
-\left[K|F_2|+\frac{\partial |F_2|}{\partial X^\mu}n^\mu\right]_{X^\mu=\tilde{X}^{\mu}}=0\\
-\left[\frac{\partial |F_2|}{\partial X^\mu}-D_\mu |F_2|\right]_{X^\mu=\tilde{X}^{\mu}}=0\\
\end{array}\right. \quad \forall \tilde{X}\in \Sigma.
\end{equation}

So it is verified that for any configuration that extremizes $\int_\Sigma d^{d}\sigma \sqrt{h} F_2$ will also extremize $\int_\Sigma d^{d}\sigma \sqrt{h} |F_2|$.
        
\end{proof}

\section{Comments on negative measures}\label{app B}

Many of the subtleties  in the program arise from the fact that $F_2$ can be negative. In this appendix, we discuss another approach we could naively take and show why it doesn't work. 
We start with the assuming the existence of both positive and negative measure threads in the optimization program itself.



Suppose that instead of optimizing $W(|F_2|)$ in Section \ref{sec: codimension 1}, we had optimized $W(F_2)$. In the case where this function is non-negative everywhere the previous derivation would be completely valid and would lead to the same optimal surface $\tilde{\Sigma}$. However, if $F_2$ is negative at some patch, an inconvenience would appear, since the optimal solutions for $W(|F_2|)$ and $W(F_2)$ might be different. This fact will become clearer later. In order to overcome this problem, we are going to define a measure $\mu$ that assigns a value of $0,\ 1$ or $-1$ to each thread. That is, $\mu:p\in \mathcal{P}\rightarrow \{-1,0,1\}$. This function can be split as the sum of two simpler ones $\mu=\mu_+-\mu_-$ such that $\mu_\pm: p\in \mathcal{P}\rightarrow \{0,1\}$. The density function does not modify its definition, but will be expressed as the sum of two densities:

\begin{equation}
\begin{split}
    \rho(x)&=\int_\mathcal{P}d\mu \int_p ds \delta(x-p(s))=\int_\mathcal{P}d\mu_+ \int_p ds \delta(x-p(s))-\int_\mathcal{P}d\mu_- \int_p ds \delta(x-p(s))\\
    &=\rho_+(x)-\rho_-(x),
\end{split}
\end{equation}

with

\begin{equation}
    \rho_\pm\equiv \int_\mathcal{P}d\mu_\pm \int_p ds \delta(x-p(s)).
\end{equation}

We will propose the following program to find the extremal surface $\Sigma$:

\begin{equation}
    \min \int_{\mathcal{P}}d\mu\ \textup{s.t. }\rho_+-\rho_-\geq F_2(x)\ \forall x\in \mathcal{M}.
\end{equation}

Here, it is important to remark that (as was explained before), although only one constraint is written, we are saying that there is one constraint per point in spacetime. This is the reason why, when the Lagrangian is proposed, an integral has to be introduced. This fact is crucial to understand why negative measures have to be introduced. Suppose that the function $F_2(x)$ is negative in a certain subset $U\subseteq \mathcal{M}$ and that we only work with positive measures. As a consequence of complementary slackness \cite{boyd_vandenberghe_2004}, it is known that, for the optimal configuration

\begin{equation}
    \lambda(x)(F_2(x)-\rho_+(x))=0\ \forall x\in \mathcal{M},
\end{equation}

and in particular

\begin{equation}
    \lambda(x)(-|F_2(x)|-\rho_+(x))=0\ \forall x\in U.
\end{equation}

Notice that for $x\in U$, either $\lambda(x)=0$ or $\rho_+=-|F_2(x)|$. By definition $\rho_+(x)$ is always non-negative so the last equality will never hold and $\lambda(x)=0$ $\forall x\in U$. At a first sight, this does not seem to be a great problem. Nonetheless, it becomes evident when $U=\mathcal{M}$ since it would imply $\lambda(x)=0$ everywhere and we do not get any surface $\Sigma$ as a solution. Now, let us write down the Lagrangian associated to our program using signed measures

\begin{equation}
\begin{split}
    L=&\int_{\mathcal{P}}d\mu_+-\int_{\mathcal{P}}d\mu_-+\int_\mathcal{M}d^dx\sqrt{-g}\lambda(x)(F_2(x)-\rho_++\rho_-)\\
    =&\int_\mathcal{M}d^dx\sqrt{-g}\lambda(x)F_2(x)+\int_{\mathcal{P}}d\mu_+\left(1-\int_p ds\lambda\right)+\int_{\mathcal{P}}d\mu_-\left(\int_p ds\lambda-1\right).
\end{split}
\end{equation}

The dual program is then

\begin{equation}
    \max \int_\mathcal{M}d^dx\sqrt{-g}\lambda(x)F_2(x),\ \textup{s.t. } \int_p ds\lambda\leq 1\ \forall p\in\mathcal{P}_+,\ \int_p ds\lambda\geq 1\ \forall p\in\mathcal{P}_-.
\end{equation}

One might think that the optimal solution takes place when $\lambda(x)$ is a delta function supported on the Cauchy slice that makes the function $\sqrt{h}F_2(x)$ maximal (see Fig \ref{negative measure threads}) such that:

\begin{equation}
    \min \int_{\mathcal{P}}d\mu=\max \int_{\Sigma} \sqrt{h}F_2(x).
\end{equation}

The problem found with this approach is the fact that, as we are considering functions that might be negative in some patches, a maximization program might not capture all the possible solutions because those that minimize the functional are not covered in this case.


\begin{figure}
    \centering
    \includegraphics[scale=0.4]{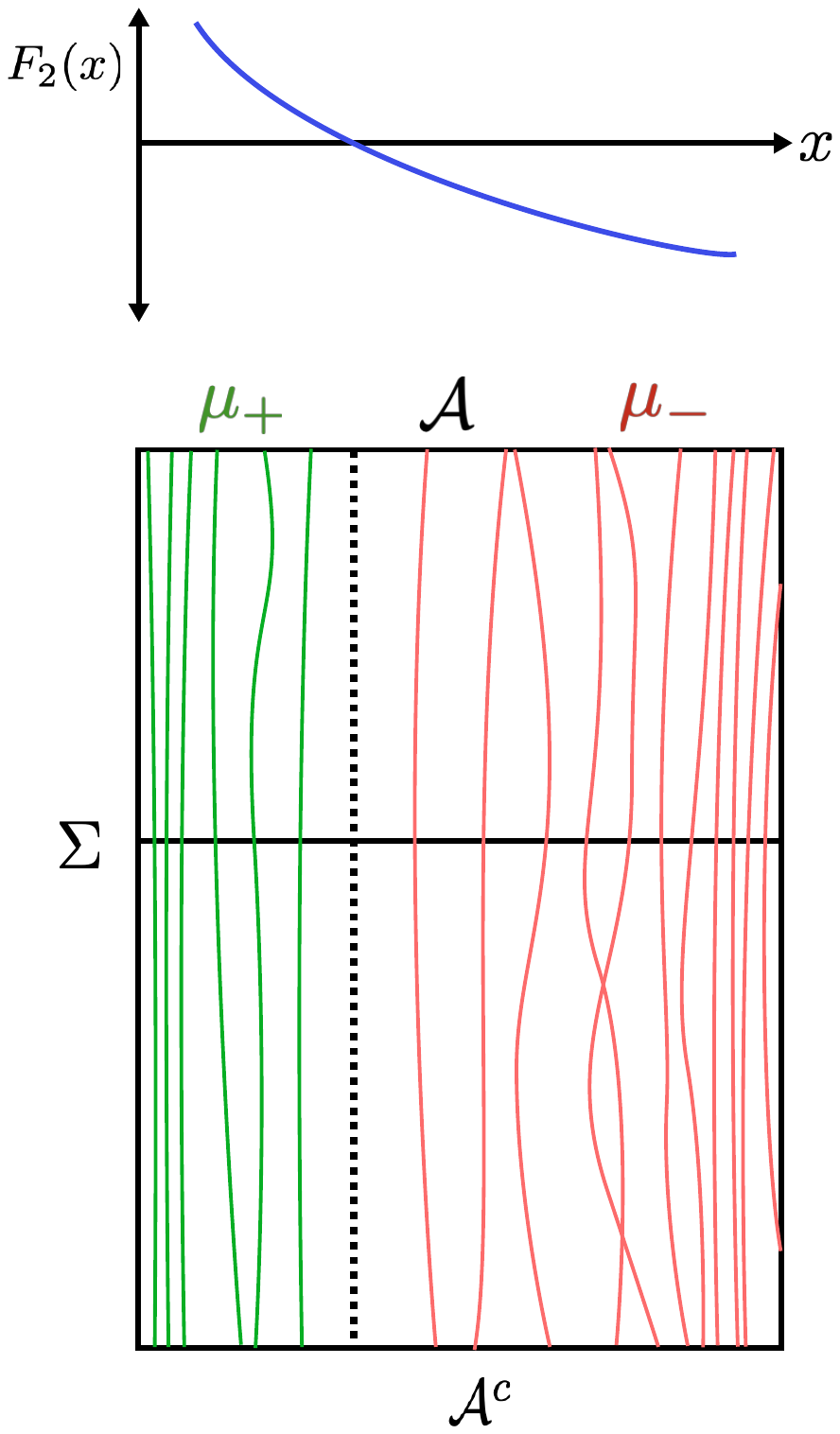}
    \caption{Negative measured threads where $F_2(x)$ is negative.}
    \label{negative measure threads}
\end{figure}

\end{appendices}
\vfill\pagebreak

\bibliographystyle{jhep}
\bibliography{references.bib}
\end{document}